\documentclass[11pt]{article}
\usepackage[left=1in,top=1in,right=1in,bottom=1in,head=.1in,nofoot]{geometry}

\usepackage{setspace,url,bm,amsmath,amssymb} % For double-spacing, URL font, math symbols

\usepackage{titlesec} % Section header formatting
\titlelabel{\thetitle.\quad} % Section header formatting
\titleformat*{\section}{\bf\large\center\uppercase} % Section header formatting

\usepackage{graphicx} % Graphics scaling
\usepackage{latexsym}
\usepackage{authblk}
\usepackage{rotating}

\usepackage{amsthm}
\theoremstyle{definition}
\newtheorem{theorem}{Theorem}
\newtheorem{lemma}{Lemma}
\newtheorem{example}{Example}

\usepackage{natbib} % ASA citation style
\bibpunct{(}{)}{;}{a}{}{,} % ASA citation style

\usepackage{etoolbox} % Bibliography underfull/overfull box fix
\apptocmd{\sloppy}{\hbadness 10000\relax}{}{} % Bibliography underfull/overfull box fix

\begin{document}
\doublespacing
\title{\bf Sharpening randomization-based causal inference for $2^2$ factorial designs with binary outcomes\footnote{Forthcoming in \emph{Statistical Methods in Medical Research}.}}
\date{\today}
\author[1]{Jiannan Lu\thanks{Address for correspondence: Jiannan Lu, One Microsoft Way, Redmond, Washington 98052-6399, U.S.A. Email: \texttt{jiannl@microsoft.com}}}
\affil[1]{Analysis and Experimentation, Microsoft Corporation}
\maketitle
\begin{abstract}

In medical research, a scenario often entertained is randomized controlled $2^2$ factorial design with a binary outcome. By utilizing the concept of potential outcomes, \cite{Dasgupta:2015} proposed a randomization-based causal inference framework, allowing flexible and simultaneous estimations and inferences of the factorial effects. However, a fundamental challenge that \cite{Dasgupta:2015}'s proposed methodology faces is that the sampling variance of the randomization-based factorial effect estimator is unidentifiable, rendering the corresponding classic ``Neymanian'' variance estimator suffering from over-estimation. To address this issue, for randomized controlled $2^2$ factorial designs with binary outcomes, we derive the sharp lower bound of the sampling variance of the factorial effect estimator, which leads to a new variance estimator that sharpens the finite-population Neymanian causal inference. We demonstrate the advantages of the new variance estimator through a series of simulation studies, and apply our newly proposed methodology to two real-life datasets from randomized clinical trials, where we gain new insights.

\end{abstract}

\section{Introduction}

Since originally introduced to conduct and analyze agricultural experiments \citep{Fisher:1935, Yates:1937}, factorial designs have been widely applied in social, behavioral and biomedical sciences, because of their capabilities to evaluate multiple treatment factors simultaneously. In particular, over the past half-century, randomized controlled $2^2$ factorial designs have become more well-adopted in medical research, in which the research interest often lies in assessing the (main and interactive) causal effects of two distinct binary treatment factors on a binary outcome. Among the lengthy list of medical studies that are powered by $2^2$ factorial designs \citep{Chalmers:1955, Hennekens:1985, Eisenhauer:1994, Rapola:1997, Franke:2000, Ayles:2008, Mhurchu:2010, Greimel:2011, Manson:2012, James:2013}, one of the most impactful examples is the landmark Physicians' Health Study \citep{Stampfer:1985}, in which over ten thousand patients were randomly assign to four experimental arms -- 1. placebo aspirin and placebo $\beta-$carotene; 2. placebo aspirin and active $\beta-$carotene; 3. active aspirin and placebo $\beta-$carotene; 4. active aspirin and active $\beta-$carotene. This study contained multiple important end-points that were binary, e.g., cardiovascular mortality. 

For traditional treatment-control studies (i.e., $2^1$ factorial designs), a well-developed and popular methodology to conduct causal inference is the potential outcomes framework \citep{Neyman:1923, Rubin:1974}, where we define causal effects as comparisons (difference, ratio, et al.) between the treated and control potential outcomes, which are assumed to be fixed for each experimental unit. Consequently, estimation and inference of causal effects solely depend on treatment assignment randomization, which is often regarded as the gold standard for causal inference \citep{Rubin:2008}. As a randomization-based methodology, the potential outcomes framework possesses several advantages against other existing approaches, many of which are model-based. For example, it is fully non-parametric and therefore more robust to model mis-specification, and better suited for finite population analyses, which under certain circumstances are more appropriate as pointed by several researchers \citep{Miller:2006}.

Realizing the salient feature of the potential outcomes framework, \cite{Dasgupta:2015} formally extended it to $2^K$ factorial designs, by defining the factorial effects as linear contrasts of potential outcomes under different treatment combinations, and proposing the corresponding estimation and inferential procedures. \cite{Dasgupta:2015} argued that by utilizing the concept of potential outcomes, the proposed randomization-based framework ``results in better understanding
of'' factorial effects, and ``allows greater flexibility in statistical inference.'' However it is worth mentioning that, while ``inherited'' many desired properties of the potential outcomes framework, inevitably it also inherited a fundamental issue -- the sampling variance of the randomization-based estimator for the factorial effects is unidentifiable, and therefore the corresponding classic ``Neymanian'' variance estimator suffers from the issue of over-estimation in general (see Section 6.5 of \cite{Imbens:2015} for a detailed discussion) -- in fact, as pointed by \cite{Aronow:2014}, it is generally impossible to unbiasedly estimate the sampling variance, because we simply cannot directly infer the association between the potential outcomes. For treatment-control studies, this problem has been extensively investigated and somewhat well-resolved, for binary \citep{Robins:1988, Ding:2016} and more general outcomes \citep{Aronow:2014}. However, to our best knowledge, similar discussions appear to be absent in the existing literature for $2^K$ factorial designs, which are of both theoretical and practical interests. Motivated by several real-life examples in medical research, in this paper we take a first step towards filling this important gap, by sharpening randomized-based causal inference for $2^2$ factorial designs with binary outcomes. To be more specific, we derive the sharp (formally defined later) lower bound of the sampling variance of the factorial effect estimator, and propose the corresponding ``improved'' Neymanian variance estimator.

The paper proceeds as follows. In Section \ref{sec:review} we briefly review the randomization-based causal inference framework for $2^2$ factorial designs, focusing on binary outcomes. Section \ref{sec:improvement} presents the bias (i.e., magnitude of over-estimation) of the classic Neymanian variance estimator, derives the sharp lower bound of the bias, proposes the corresponding improved Neymanian variance estimator, and illustrate our results through several numerical and visual examples. Sections \ref{sec:simu} conducts a series of simulation studies to highlight the performance of the improved variance estimator. Section \ref{sec:application} applied our newly proposed methodology to two real-life medical studies, where new insights are gained. Section \ref{sec:conclusion} discusses future directions and concludes. We relegate the technical details to Appendices \ref{sec:proofs} and \ref{sec:simu-add}.

\section{Randomization-based causal inference for $2^2$ factorial designs with binary outcomes}\label{sec:review}

\subsection{$2^2$ factorial designs}

To review Neymanian causal inference for $2^2$ factorial designs, we adapt materials by \cite{Dasgupta:2015} and \cite{Lu:2016a}, and tailor them to the specific case with binary outcomes. In $2^2$ factorial designs, there are two treatment factors (each with two-levels coded as -1 and 1) and $4$ distinct treatment combinations $\bm z_j$ $(j=1, \ldots, 4).$ To define them, we rely on the model matrix \citep{Wu:2009} 
\begin{equation*}
\bm H =
\bordermatrix{& \bm h_0 & \bm h_1& \bm h_2 & \bm h_3\cr
              & 1  & -1 &  -1  & 1 \cr
              & 1 & -1 &  1  & -1 \cr
              & 1  & 1 &  -1  & -1 \cr
              & 1 & 1 & 1 & 1}.
\end{equation*}
The treatment combinations are $\bm z_1=(-1, -1),$ $\bm z_2=(-1, 1),$ $\bm z_3=(1, -1)$ and $\bm z_4=(1, 1),$ and later we will use $\bm h_1,$ $\bm h_2$ and $\bm h_3$ to define the factorial effects.

\subsection{Randomization-based inference}

By utilizing potential outcomes, \cite{Dasgupta:2015} proposed a framework for randomization-based causal inference for $2^K$ factorial designs. For our purpose, we consider a $2^2$ factorial design with $N \ge 8$ experimental units. Under the Stable Unit Treatment Value Assumption \citep{Rubin:1980}, for $i = 1, \ldots, N,$ we define $Y_i(\bm z_j)$ as the potential outcome of unit $i$ under treatment combination $\bm z_j,$ and let
$
\bm Y_i = \{ Y_i(\bm z_1), Y_i(\bm z_2), Y_i(\bm z_3), Y_i(\bm z_4) \}^\prime.
$
In this paper we only consider binary outcomes, i.e., $Y_i(\bm z_j) \in \{0, 1\}$ for all $i = 1, \ldots, N$ and $j = 1, \ldots, 4.$

To save space, we introduce two sets of notations. First, we let
\begin{equation*}
D_{k_1, k_2, k_3, k_4} = \sum_{i=1}^N \prod_{j=1}^4 1_{\{Y_i(\bm z_j) = k_j\}}
\quad
(k_1, k_2, k_3, k_4 \in \{0, 1\}), 
\end{equation*}
Consequently, instead of specifying the potential outcomes
$
(\bm Y_1, \ldots, \bm Y_N)
$
entry by entry, we can equivalently characterize them using the ``joint distribution'' vector
$
(D_{0000}, D_{0001}, \ldots, D_{1110}, D_{1111}),
$
where the indices are ordered binary representations of zero to fifteen. Second, for all non-empty sets
$
\{j_1, \ldots, j_s\} \subset \{1, \ldots, 4\},
$
we let
\begin{equation*}
N_{j_1, \ldots, j_s} 
= \sum_{i=1}^N 1_{\left\{Y_i(\bm z_{j_1}) = 1, \ldots, Y_i(\bm z_{j_s}) = 1\right\}}.
\end{equation*}
Therefore, for $j=1, \ldots, 4,$ the average potential outcome for $\bm z_j$ is
\begin{equation*}
p_j = N^{-1} \sum_{i = 1}^N Y_i (\bm z_j) = N_j / N,
\end{equation*}
and let
$
\bm p = ( p_1, p_2, p_3, p_4 )^\prime.
$
Define the $l$th (individual and population) factorial effects as
\begin{equation}
\label{eq:factorial-effects}
\tau_{il} = 2^{-1} \bm h_l^\prime \bm Y_i
\quad
(i = 1, \ldots, N);
\quad
\bar \tau_l = 2^{-1} \bm h_l^\prime \bm p,
\end{equation}
for $l=1, 2, 3,$ which correspond to the main effects of the first and second treatment factors, and their interaction effect, respectively.

Having defined the treatment combinations, potential outcomes and factorial effects, next we discuss the treatment assignment and observed data. Suppose for $j=1, \ldots, 4,$ we randomly assign $n_j$ (a pre-specified constant) units to treatment combination $\bm z_j.$ Let
\begin{equation*}
W_i(\bm z_j)
= 
\begin{cases}
1, & \text{if unit $i$ is assigned to $\bm z_j,$ } \\
0, & \text{otherwise;} \\
\end{cases}
\quad
(i=1, \ldots, N; j=1, \ldots, 4)
\end{equation*}
be the treatment assignments, and
\begin{equation*}
Y_i^\textrm{obs} = \sum_{j=1}^4 W_i(\bm z_j) Y_i(\bm z_j)
\end{equation*}
be the observed outcome for unit $i,$ and
\begin{equation*}
n_j^\textrm{obs} 
= \sum_{i=1}^N W_i(\bm z_j)Y_i(\bm z_j)
= \sum_{i: W_i(\bm z_j) = 1} Y_i^\textrm{obs}.
\end{equation*}
Therefore, the average observed potential outcome for $\bm z_j$ is 
$
\hat p_j = n_j^\mathrm{obs} / n_j,
$
for all
$
j = 1, \ldots, 4.
$
Denote
$
\hat{\bm p}
= 
(
\hat p_1, 
\hat p_2,
\hat p_3,
\hat p_4
)^\prime,
$
and the randomization-based estimators for 
$
\bar \tau_l
$
is
\begin{equation}
\label{eq:factorial-effects-estimator}
\hat {\bar \tau}_l =  2^{-1} \bm h_l^\prime \hat{\bm p}
\quad
(l = 1, 2, 3),
\end{equation}
which are unbiased with respect to the randomization distribution. 

Motivated by several relevant discussions in the existing literature \citep{Freedman:2008, Lin:2013, Dasgupta:2015, Ding:2016, Ding:2017}, \cite{Lu:2016a, Lu:2016b} proved the consistency and asymptotic Normality of the randomization-based estimator in \eqref{eq:factorial-effects-estimator}, and derived its sampling variance as
\begin{equation}
\label{eq:factorial-effects-variance}
\mathrm{Var}(\hat {\bar \tau}_l) = \frac{1}{4} \sum_{j=1}^4  S_j^2 / n_j - \frac{1}{N} S^2(\bar \tau_l),
\end{equation}
where for $j=1, \ldots, 4$
\begin{equation*}
S_j^2 
= 
(N-1)^{-1}
\sum_{i=1}^N 
\{ 
Y_i(\bm z_j) - p_j
\}^2 
= \frac{N}{N-1} p_j (1 - p_j)
\end{equation*}
is the variance of potential outcomes for $\bm z_j,$ and
\begin{equation*}
S^2(\bar \tau_l) 
= 
(N-1)^{-1}
\sum_{i=1}^N 
(\tau_{il} - \bar \tau_l)^2
\end{equation*}
is the variance of the $l$th (individual) factorial effects in \eqref{eq:factorial-effects}.

\section{Improving the Neymanian variance estimator} \label{sec:improvement}

\subsection{Background}

Given the sampling variance in \eqref{eq:factorial-effects-variance}, we estimate it by substituting $S_j^2$ with its unbiased estimate
\begin{equation*}
s_j^2 
= 
(n_j - 1)^{-1} 
\sum_{i: W_i(\bm z_j) = 1}
\{
Y_i^{\textrm{obs}} - \hat p_j 
\}^2
= \frac{n_j}{n_j - 1} \hat p_j (1 - \hat p_j),
\end{equation*}
and substituting 
$
S^2(\bar \tau_l) 
$
with its lower bound 0 (due to the fact that it is not identifiable, because none of the individual factorial effects $\tau_{il}$'s are observable). Consequently, we obtain the ``classic Neymanian'' variance estimator is
\begin{equation}
\label{eq:factorial-effects-variance-estimator}
\widehat{\mathrm{Var}}_{\mathrm{N}}(\hat {\bar \tau}_l) 
= \frac{1}{4} \sum_{j=1}^4 s_j^2 / n_j
= \frac{1}{4} \sum_{j=1}^4 \frac{\hat p_j (1 - \hat p_j)}{n_j - 1}.
\end{equation}
This estimator over-estimates the true sampling variance on average by
\begin{equation}
\label{eq:factorialeffect-variance-estimator-over}
\mathrm{E} 
\left\{ 
\widehat{\mathrm{Var}}_{\mathrm{N}}(\hat {\bar \tau}_l) 
\right\} 
- \mathrm{Var}(\hat{\bar \tau}_l) 
= S^2(\bar \tau_l) / N,
\end{equation}
unless strict additivity \citep{Dasgupta:2015} holds, i.e., 
\begin{equation*}
\tau_{il} = \tau_{i^\prime l}
\quad
(\forall i,i^\prime = 1, \ldots, N),
\end{equation*}
which is unlikely to happen in real-life scenarios, especially for binary outcomes \citep{LaVange:2005, Rigdon:2015}. We summarize and illustrate the above results by the following example.

\begin{example} 
Consider a hypothetical $2^2$ factorial design with $N=16$ units, whose potential outcomes, factorial effects and summary statistics are shown in Table \ref{tab:potential-outcomes}, from which we draw several conclusions -- first, the population-level factorial effects in \eqref{eq:factorial-effects} are -0.1563, -0.0313 and -0.0313, respectively; second, the sampling variances of the randomization-based estimators in \eqref{eq:factorial-effects-estimator} are 0.0425, 0.0493 and 0.0493, respectively; third, if we employ the classic Neymanian variance estimator in \eqref{eq:factorial-effects-variance-estimator}, on average we will over-estimate the sampling variances by 52.5\%, 31.6\% and 31.6\%, respectively.

\begin{table}[htbp]
\centering
\caption{\small Hypothetical potential outcomes, factorial effects and summary statistics for a $2^2$ factorial design with 16 experimental units.}
\vspace{2mm}
\small
\begin{tabular}{c|cccc|ccc}
\hline
\multicolumn{1}{c}{Unit}   & \multicolumn{4}{|c|}{Potential outcomes} & \multicolumn{3}{c}{Factorial Effects} \\
 ($i$) &      $Y_i(\bm z_1)$   &    $Y_i(\bm z_2)$     &      $Y_i(\bm z_3)$     &   $Y_i(\bm z_4)$          &  $\tau_{i1}$   &  $\tau_{i2}$  & $\tau_{i3}$  \\ 
 \hline
  1 &  1 &  1 &  1 &  0 & -0.5 & -0.5 & -0.5 \\ 
  2 &  0 &  0 &  1 &  1 & 1.0 & 0.0 & 0.0 \\ 
  3 &  1 &  1 &  0 &  0 & -1.0 & 0.0 & 0.0 \\ 
  4 &  1 &  0 &  1 &  0 & 0.0 & -1.0 & 0.0 \\ 
  5 &  0 &  1 &  0 &  0 & -0.5 & 0.5 & -0.5 \\ 
  6 &  1 &  0 &  0 &  1 & 0.0 & 0.0 & 1.0 \\ 
  7 &  0 &  1 &  0 &  0 & -0.5 & 0.5 & -0.5 \\ 
  8 &  1 &  1 &  0 &  1 & -0.5 & 0.5 & 0.5 \\ 
  9 &  0 &  1 &  1 &  0 & 0.0 & 0.0 & -1.0 \\ 
  10 &  0 &  0 &  1 &  1 & 1.0 & 0.0 & 0.0 \\ 
  11 &  1 &  1 &  0 &  0 & -1.0 & 0.0 & 0.0 \\ 
  12 &  1 &  0 &  0 &  0 & -0.5 & -0.5 & 0.5 \\ 
  13 &  0 &  1 &  0 &  1 & 0.0 & 1.0 & 0.0 \\ 
  14 &  0 &  0 &  0 &  0 & 0.0 & 0.0 & 0.0 \\ 
  15 &  1 &  1 &  1 &  0 & -0.5 & -0.5 & -0.5 \\ 
  16 &  1 &  0 &  1 &  1 & 0.5 & -0.5 & 0.5 \\ 
\hline
Mean  & $p_1$ & $ p_2$ & $p_3$ & $p_4$ & $\bar \tau_1$ & $\bar \tau_2$ & $\bar \tau_3$ \\
 & = 0.5625 & = 0.5625 & = 0.4375 & = 0.3750 & = -0.1563 & = -0.0313 & = -0.0313          \\
 \hline
Variance  & $S_1^2$ & $S_2^2$ & $S_3^2$ & $S_4^2$ & $S^2(\bar \tau_1)$ & $S^2(\bar \tau_2)$ & $S^2(\bar \tau_3)$ \\
 & = 0.2625 & = 0.2625 & = 0.2625 & = 0.2500 & = 0.3573 & = 0.2490 & = 0.2490       \\ 
\hline
\end{tabular}
\label{tab:potential-outcomes}
\end{table}
\end{example}

\subsection{Sharp lower bound of the sampling variance}

As demonstrated in previous sections, the key to improve the classic Neymanian variance estimator \eqref{eq:factorial-effects-variance-estimator} is obtaining a non-zero and identifiable lower bound of
$
S^2(\bar \tau_l).
$
To achieve this goal, we adopt the partial identification philosophy, commonly used in the existing literature to bound either the randomization-based sampling variances of causal parameters \citep{Aronow:2014}, or the causal parameters themselves \citep{Zhang:2003, Fan:2010, Lu:2015}.

We first present two lemmas, which play central roles in the proof of our main theorem. 

\begin{lemma}
\label{lemma:variance-taul-formula}
Let 
$
\bm h_l = (h_{1l}, h_{2l}, h_{3l}, h_{4l})^\prime
$
for all $l=1, 2, 3.$ Then
\begin{equation*}
S^2(\bar \tau_l) 
= 
\frac{1}{4(N-1)} 
\left(
\sum_{j=1}^4 N_j
+
\sum_{j \ne j^\prime} h_{lj}h_{lj^\prime} N_{jj^\prime}
\right)
-
\frac{N}{N-1} \bar \tau_l^2.
\end{equation*}
\end{lemma}

\begin{lemma}
\label{lemma:ie-inequality}
For all $l = 1, 2, 3,$ define
\begin{equation*}
\bm J_{l-} = \{j: h_{lj} = -1\},
\quad
\bm J_{l+} = \{j: h_{lj} = 1\}.
\end{equation*}
Then
\begin{equation}
\label{eq:ie-inequality}
\sum_{j=1}^4 N_j
+
\sum_{j \ne j^\prime} h_{lj}h_{lj^\prime} N_{jj^\prime}
\ge
\left|
\sum_{j=1}^4 h_{lj} N_l
\right|,
\end{equation}
and the equality in \eqref{eq:ie-inequality} holds if and only if
\begin{equation}
\label{eq:ie-inequality-hold-1}
\sum_{j \in \bm J_{l+}} Y_i(\bm z_j) - 1
\le
\sum_{j \in \bm J_{l-}} Y_i(\bm z_j) 
\le 
\sum_{j \in \bm J_{l+}} Y_i(\bm z_j)
\quad
(i = 1, \ldots, N)
\end{equation}
or
\begin{equation}
\label{eq:ie-inequality-hold-2}
\sum_{j \in \bm J_{l+}} Y_i(\bm z_j)
\le
\sum_{j \in \bm J_{l-}} Y_i(\bm z_j) 
\le 
\sum_{j \in \bm J_{l+}} Y_i(\bm z_j) + 1
\quad
(i = 1, \ldots, N).
\end{equation}
\end{lemma}

We provide the proofs of Lemmas \ref{lemma:variance-taul-formula} and \ref{lemma:ie-inequality} in Appendix \ref{sec:proofs}. With the help of the lemmas, we present an identifiable sharp lower bound of
$
S^2(\bar \tau_l).
$

\begin{theorem}
\label{thm:variance-taul-lower-bound}
The sharp lower bound for 
$
S^2(\bar \tau_l)
$
is 
\begin{equation}
\label{eq:variance-taul-lower-bound}
S_{\mathrm{LB}}^2(\bar \tau_l)
=
\frac{N}{N-1}
\max 
\left\{
|\bar \tau_l| ( 1/2 - |\bar \tau_l|), 0
\right\}.
\end{equation}
The equality in \eqref{eq:variance-taul-lower-bound} holds if and only if \eqref{eq:ie-inequality-hold-1} or \eqref{eq:ie-inequality-hold-2} holds.
\end{theorem}

\medskip
By employing the inclusion-exclusion principle and Bonferroni's inequality, we provide the proof of Theorem \ref{thm:variance-taul-lower-bound} in Appendix \ref{sec:proofs}. The lower bound in Theorem \ref{thm:variance-taul-lower-bound} is sharp in the sense that it is compatible with the marginal counts of the potential outcomes 
$
(N_1, N_2, N_3, N_4)
$
(and consequently $ \bar \tau_l$). To be more specific, for fixed values of 
$
(N_1, N_2, N_3, N_4),
$
there exists a hypothetical set of potential outcomes
$
(\bm Y^*_1, \ldots, \bm Y^*_N),
$
such that 
\begin{equation*}
\sum_{i=1}^N Y_i^*(\bm z_j) = N_j
\quad
(j=1, \ldots, 4);
\quad
S^{*2}(\bar \tau_l) = S_{\mathrm{LB}}^2(\bar \tau_l).
\end{equation*}

Theorem \ref{thm:variance-taul-lower-bound} effectively generalizes the discussions regarding binary outcomes by \cite{Robins:1988} and \cite{Ding:2016}, from treatment-control studies to $2^2$ factorial designs. In particular, the conditions in \eqref{eq:ie-inequality-hold-1} and \eqref{eq:ie-inequality-hold-2} echo the parallel results by \cite{Ding:2016}, and therefore we name them the ``generalized'' monotonicity conditions on the potential outcomes. However, intuitive and straightforward as it seems, proving Theorem \ref{thm:variance-taul-lower-bound} turns out to be a non-trivial task.

\subsection{The ``improved'' Neymanian variance estimator}

The sharp lower bound in \eqref{eq:variance-taul-lower-bound} leads to the ``improved'' Neymanian variance estimator
\begin{equation}
\label{eq:factorial-effects-variance-estimator-improved}
\widehat{\mathrm{Var}}_{\mathrm{IN}}(\hat {\bar \tau}_l) 
= \underbrace{
\frac{1}{4} \sum_{j=1}^4 \frac{\hat p_j (1 - \hat p_j)}{n_j - 1}
}_{
\widehat{\mathrm{Var}}_{\mathrm{N}}(\hat {\bar \tau}_l) 
}
-
\underbrace{
\vphantom {
\frac{1}{4} \sum_{j=1}^4 \frac{\hat p_j (1 - \hat p_j)}{n_j - 1}
} 
\frac{1}{N-1}
\max 
\left\{
|\hat{\bar \tau}_l| ( 1/2 - |\hat{\bar \tau}_l|), 0
\right\}
}_{
\hat S_{\mathrm{LB}}^2(\bar \tau_l) / N
},
\end{equation}
which is guaranteed to be smaller than the classic Neymanian variance estimator in \eqref{eq:factorial-effects-variance-estimator} for any observed data, because the correction term on the right hand side of \eqref{eq:factorial-effects-variance-estimator-improved} is always non-negative. For example, for balanced designs (i.e., $n_1 = n_2 = n_3 = n_4$) with large sample sizes, the relative estimated variance reduction is
\begin{equation*}
\gamma_l = 
\frac{
\hat S_{\mathrm{LB}}^2(\bar \tau_l) / N
}
{
\widehat{\mathrm{Var}}_{\mathrm{N}}(\hat {\bar \tau}_l) 
}
\approx
\frac{|\hat{\bar \tau}_l| ( 1/2 - |\hat{\bar \tau}_l|)}{\sum_{j=1}^4 \hat p_j (1 - \hat p_j)}.
\end{equation*}
We illustrate the above results by the following numerical example.

\begin{example}
Consider a balanced $2^2$ factorial design with $N = 400$ experimental units, so that 
$
(n_1, n_2, n_3, n_4) = (100, 100, 100, 100).
$
For the purpose of visualizing the estimated variance reduction under various settings, we repeatedly draw  
\begin{equation*}
n_j^\mathrm{obs} \stackrel{iid.}{\sim} \lfloor \mathrm{Unif}(0, 100) \rfloor
\quad
(j=1, \ldots, 4)
\end{equation*} 
for 5000 times, and plot the corresponding $\gamma_1$'s in Figure \ref{fg:gamma}. We can draw several conclusions from the results. First, for 13\% of the times $\gamma_1$ is smaller than 1\%, corresponding to cases where $\hat{\bar \tau}_l \approx -0.5,$ $0$ or $0.5.$ Second, for 13\% of the times $\gamma_1$ is larger than 10\%. Third, the largest $\gamma_1$ is approximately 20.5\%, corresponding to the case where
$
(n_1^\mathrm{obs}, n_2^\mathrm{obs}, n_3^\mathrm{obs}, n_4^\mathrm{obs})
= (0, 0, 16, 14)
$ 
and
$
\hat {\bar \tau}_1 = 0.15.
$
\end{example}

\begin{figure}
\centering
\includegraphics[width = .7 \linewidth]{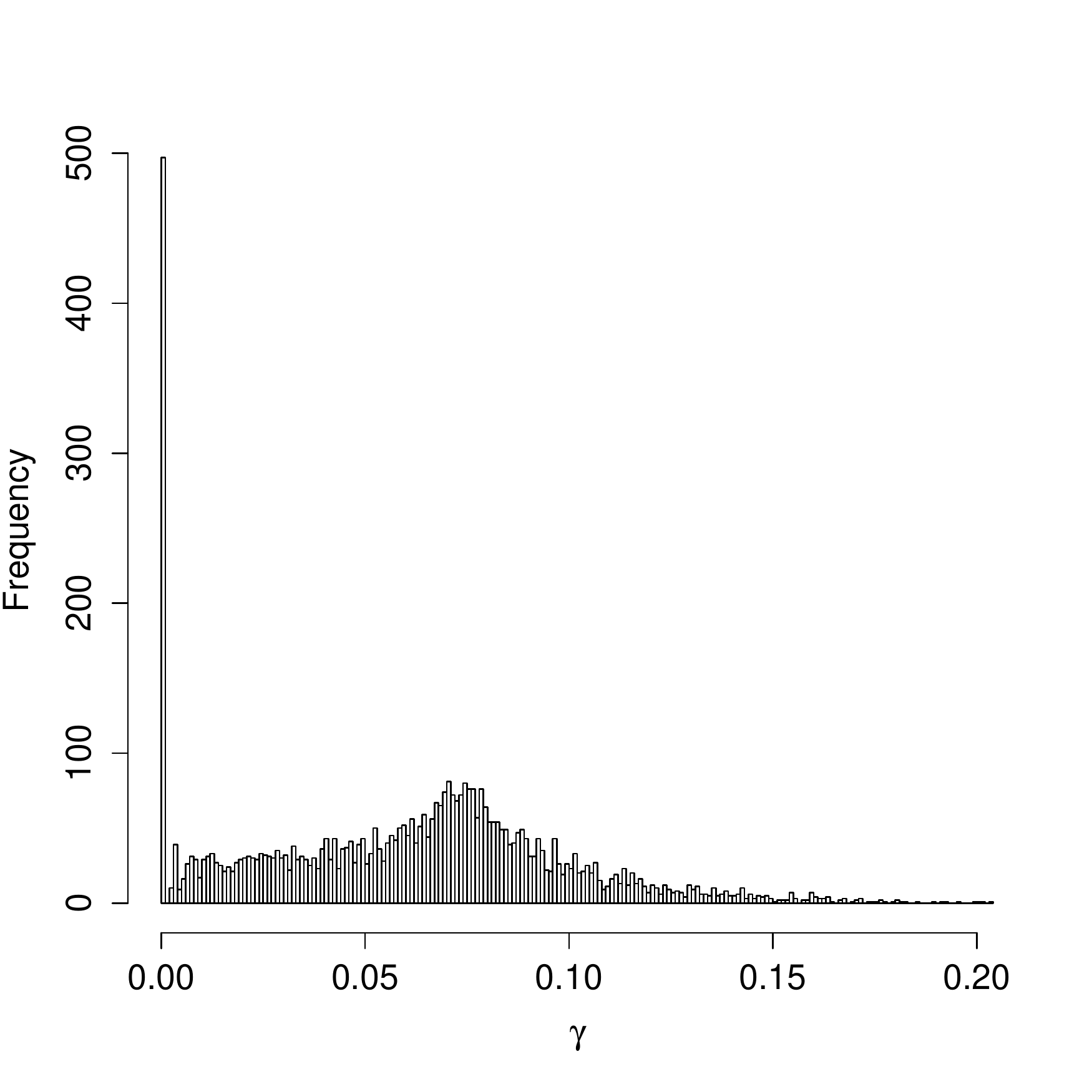}
\caption{Visualization of the relative variance estimation reduction for $\hat{\bar \tau}_1,$ through 5000 repeated samplings of $(n_1^\mathrm{obs}, n_2^\mathrm{obs}, n_3^\mathrm{obs}, n_4^\mathrm{obs})$ using the Uniform distribution.}
\label{fg:gamma}
\end{figure}

\medskip
As pointed out by several researchers \citep{Aronow:2014, Ding:2016}, the probabilistic consistency of the factorial effect estimator
$
\hat{\bar \tau}_l
$
guarantees that the improved Neymanian variance estimator still over-estimates the sampling variance on average, unless one of the generalized monotonicity conditions in \eqref{eq:ie-inequality-hold-1}--\eqref{eq:ie-inequality-hold-2} holds. Nevertheless, it does improve the classic Neymanian variance estimator in \eqref{eq:factorial-effects-variance-estimator}, and more importantly, this improvement is the ``best we can do'' without additional information. In the next section, we conduct simulation studies to demonstrate the finite-sample properties of, and to compare the performances of, the classic and improved Neymanian variance estimators.

\section{Simulation studies}
\label{sec:simu}

To save space, we focus on the first factorial effect
$
\bar \tau_1
$
and its randomization-based statistical inference. To mimic the empirical examples that we will re-analyze in the next section, we choose the sample size $N=800.$ Moreover, to (at least to some extent) explore the complex dependence structure of the potential outcomes, we adopt the latent multivariate Normal model for the underlying data generation mechanism is. To be more specific, let
\begin{equation*}
\bm \eta_i 
= 
\{
\eta_i(\bm z_1), \eta_i(\bm z_2), \eta_i(\bm z_3), \eta_i(\bm z_4) 
\}^\prime
\stackrel{iid.}{\sim} N \left( \bm \mu, \bm \Sigma \right)
\quad
(i = 1, \ldots, N),
\end{equation*}
and assume that for each $i$ 
\begin{equation*}
Y_i(\bm z_j) = 1_{\{\eta_i(\bm z_j) \ge 0\}}
\quad
(j = 1, \ldots, 4).
\end{equation*}
We consider the following six cases:
\begin{equation*}
\bm \mu = 
\underbrace{
\overbrace{(-2, -2, -2, -2)}^{\mathrm{Case} \: 1}, 
\; 
\overbrace{(0, 0, 0, 0)}^{\mathrm{Case} \: 2}
}_{\bar \tau_1 \approx 0}, 
\quad 
\underbrace{
\overbrace{(0, 0, 0, 2)}^{\mathrm{Case} \: 3}, 
\; 
\overbrace{(-2, -2, 0, -2)}^{\mathrm{Case} \: 4}
}_{\bar \tau_1 \approx 0.25}, 
\quad
\underbrace{
\overbrace{(-2, -2, -2, 2)}^{\mathrm{Case} \: 5}, 
\; 
\overbrace{(-2, 0, 0, 2)}^{\mathrm{Case} \: 6}
}_{\bar \tau_1 \approx 0.5}.
\end{equation*}
We choose the above values for $\bm \mu,$ so that the corresponding factorial effects $\bar \tau_1 \approx 0$ (the approximaition is due to finite-sample fluctuation) for Cases 1--2. Similarly, $\bar \tau_1 \approx 0.25$ for Cases 3 and 4, and $\bar \tau_1 \approx 0.5$ for Cases 5 and 6. Therefore, we can examine the scenarios where the sharp lower bound
$
S_{\mathrm{LB}}^2(\bar \tau_l)
$
in \eqref{eq:factorial-effects-variance-estimator-improved} are either small or large in magnitude. Moreover, we partially adopt the simulation settings by \cite{Dasgupta:2015} and let
\begin{equation*}
\bm \Sigma =
\renewcommand{\arraystretch}{.85}
\left( 
\begin{array}{cccc}
1 & \rho & \rho & \rho \\
\rho & 1 & \rho & \rho \\
\rho & \rho & 1 & \rho \\
\rho & \rho & \rho & 1
\end{array} 
\right)
\quad
(\rho = -1 /3, 0, 1/2),
\end{equation*}
which corresponds to negatively correlated, independent and positively correlated potential outcomes, respectively. The aforementioned data generation mechanism resulted eighteen ``joint distributions'' of the potential outcomes
$
(D_{0000}, D_{0001}, \ldots, D_{1111}),
$
which we report in the third column of Table \ref{tab:simu}. For each simulation case (i.e., row of Table \ref{tab:simu}), we adopt the following three-step procedure:
\begin{enumerate}
\item Use \eqref{eq:factorial-effects} and \eqref{eq:factorial-effects-variance} to calculate
$
\bar \tau_1,
$
the sampling variance of its randomization-based estimator and its lower bound, respectively, and report them in the fourth to sixth columns of Table \ref{tab:simu};

\item Independently draw $10000$ treatment assignments from a balanced $2^2$ factorial design with $(n_1, n_2, n_3, n_4) = (200, 200, 200, 200);$

\item For each observed data-set, use \eqref{eq:factorial-effects-estimator}, \eqref{eq:factorial-effects-variance-estimator} and \eqref{eq:factorial-effects-variance-estimator-improved} to calculate the point estimate of $\bar \tau_1,$ the classic and improved Neymanian variance estimates respectively, based on which we construct two 95\% confidence intervals. 
\end{enumerate}
To examine the performances of the classic and improved Neymanian variance estimators in \eqref{eq:factorial-effects-variance-estimator} and \eqref{eq:factorial-effects-variance-estimator-improved}, in the last six columns of Table \ref{tab:simu}, we report the relative (i.e., percentage wise) over-estimations of the true sampling variance, the average lengths and the coverage rates of their corresponding confidence intervals of the two estimators, respectively.

We can draw several conclusions from the results. First, because of the non-negative correction term
$
\hat S_{\mathrm{LB}}^2(\bar \tau_l) / N,
$
for all cases the improved Neymanian variance estimator \eqref{eq:factorial-effects-variance-estimator-improved} reduces the over-estimation of the sampling variance, shortens the confidence intervals and achieves better coverage rates without under-covering. For example, in Case 4 with $\rho = 1/2,$ the improved Neymanian variance estimator reduces the coverage rate from 0.974 to 0.956, achieving near nominal level. Second, by comparing Case 1 with Case 2 (or 3 with 4, 5 with 6), we can see that for a fixed $\bar \tau_l,$ although the absolute magnitude of the correction term is the same, the performance (i.e., reduction of percentage of over-estimation, average length and coverage rate) of the improved Neymanian variance estimator might differ significantly, depending on the ``marginal distributions'' of the potential outcomes (characterized by the mean parameter $\bm \mu$). Third, for a fixed marginal distribution, the performance of the improved Neymanian variance estimator might also differ significantly, depending on the dependence structure of the potential outcomes (characterized by the association parameter $\rho$). Fourth, in certain scenarios, while the improved Neymanian variance estimator only slightly shortens the confidence interval, it leads to a non-ignorable improvement on coverage rates. For example, in Case 5 with $\rho = 0,$ a less than 5\% shorter confidence interval reduces the coverage rate from 0.976 to 0.966.

To take into account alternative data generation mechanisms and thus provide a more comprehensive pircute, in Appendix \ref{sec:simu-add} we conduct an additional series of simulation studies, where we focus on several discrete outcome distributions. The results largely agree with the above conclusions.

\begin{sidewaystable}[htbp]
\caption{\footnotesize Simulation study results. The first three columns contain the case label, the value of the association parameter $\rho$ and the corresponding joint distribution of the potential outcomes. The next three columns contain the true values of the first factorial effect, the sampling variance of the individual factorial effects $\tau_{i1}$ $(i=1, \ldots, N),$ and its sharp lower bound. To examine the performances of the classic and improved Neymanian variance estimators in \eqref{eq:factorial-effects-variance-estimator} and \eqref{eq:factorial-effects-variance-estimator-improved}, the last six columns contain their percentages of over-estimation of the true sampling variance, and the average lengths and coverage rates of their corresponding confidence intervals. 
}
\vspace{2mm}
\centering
\scriptsize
\begin{tabular}{cccccccccccc}
  \hline
Case & $\rho$ & $(D_{0000}, D_{0001}, \ldots, D_{1110}, D_{1111})$ & $\bar \tau_1$ & $S^2(\bar \tau_1)$ & $S_\mathrm{LB}^2(\bar \tau_1)$ & $\mathrm{voe}_N$ & $\mathrm{voe}_I$ & $\mathrm{len}_N$ & $\mathrm{len}_I$ & $\mathrm{cover}_N$ & $\mathrm{cover}_I$ \\ 
  \hline
1 & -1/3 & (723, 14, 21, 1, 20, 0, 0, 0, 21, 0, 0, 0, 0, 0, 0, 0) & -0.003 & 0.025 & 0.001 & 34.7\% & 29.2\% & 0.022 & 0.021 & 0.977 & 0.969 \\ 

1 & 0 & (726, 18, 15, 1, 19, 0, 1, 0, 19, 0, 1, 0, 0, 0, 0, 0) & -0.002 & 0.023 & 0.001 & 31.3\% & 25.9\% & 0.022 & 0.021 & 0.975 & 0.966 \\ 

1 & 1/2 & (740, 10, 16, 0, 8, 1, 1, 1, 9, 4, 3, 0, 3, 0, 3, 1) & 0.001 & 0.018 & 0.000 & 21.7\% & 16.9\% & 0.022 & 0.022 & 0.972 & 0.961 \\ [1.5mm]

2 & -1/3 & (0, 29, 29, 93, 35, 79, 68, 28, 44, 96, 93, 36, 82, 42, 46, 0) & -0.014 & 0.309 & 0.007 & 44.9\% & 43.1\% & 0.071 & 0.070 & 0.979 & 0.977 \\ 

2 & 0 & (44, 61, 43, 52, 48, 46, 47, 59, 44, 46, 55, 51, 46, 56, 52, 50) & 0.016 & 0.252 & 0.008 & 33.7\% & 31.9\% & 0.071 & 0.070 & 0.979 & 0.972 \\

2 & 1/2 & (182, 41, 42, 23, 41, 27, 30, 34, 33, 22, 24, 39, 26, 40, 38, 158) & -0.001 & 0.158 & 0.001 & 18.7\% & 17.3\% & 0.071 & 0.070 & 0.965 & 0.958 \\ [1.5mm]

3 & -1/3 & (0, 34, 0, 117, 0, 143, 1, 116, 0, 110, 5, 113, 4, 117, 7, 33) & 0.228 & 0.220 & 0.062 & 40.0\% & 28.8\% & 0.062 & 0.060 & 0.975 & 0.970 \\ 

3 & 0 & (2, 118, 2, 91, 5, 95, 4, 77, 4, 97, 1, 112, 0, 100, 0, 92) & 0.239 & 0.188 & 0.062 & 32.1\% & 21.6\% & 0.062 & 0.060 & 0.976 & 0.970 \\ 

3 & 1/2 & (20, 177, 3, 66, 0, 68, 0, 75, 2, 61, 0, 60, 0, 62, 0, 206) & 0.239 & 0.144 & 0.062 & 22.6\% & 12.9\% & 0.062 & 0.060 & 0.970 & 0.964 \\ [1.5mm]

4 & -1/3 & (340, 19, 386, 6, 14, 0, 7, 0, 23, 0, 5, 0, 0, 0, 0, 0) & 0.237 & 0.089 & 0.062 & 35.5\% & 10.7\% & 0.041 & 0.037 & 0.978 & 0.960 \\ 

4 & 0 & (371, 6, 381, 11, 9, 1, 5, 0, 10, 0, 6, 0, 0, 0, 0, 0) & 0.244 & 0.081 & 0.062 & 35.4\% & 8.2\% & 0.039 & 0.035 & 0.976 & 0.958 \\ 

4 & 1/2 & (424, 2, 331, 13, 4, 0, 10, 1, 1, 0, 10, 0, 0, 0, 3, 1) & 0.220 & 0.075 & 0.062 & 31.6\% & 5.6\% & 0.039 & 0.035 & 0.974 & 0.956 \\ [1.5mm]

5 & -1/3 & (15, 734, 0, 13, 3, 18, 0, 0, 2, 15, 0, 0, 0, 0, 0, 0) & 0.472 & 0.025 & 0.013 & 38.0\% & 16.9\% & 0.021 & 0.019 & 0.977 & 0.965 \\ 

5 & 0 & (20, 719, 0, 23, 2, 20, 0, 0, 0, 16, 0, 0, 0, 0, 0, 0) & 0.477 & 0.027 & 0.011 & 34.8\% & 20.1\% & 0.022 & 0.021 & 0.976 & 0.966 \\ 

5 & 1/2 & (19, 713, 0, 20, 0, 17, 0, 2, 0, 18, 0, 4, 0, 6, 0, 1) & 0.471 & 0.030 & 0.014 & 31.9\% & 16.9\% & 0.025 & 0.023 & 0.977 & 0.967 \\ [1.5mm]

6 & -1/3 & (0, 148, 4, 234, 3, 242, 14, 140, 0, 10, 2, 0, 0, 2, 1, 0) & 0.471 & 0.164 & 0.014 & 42.6\% & 39.0\% & 0.052 & 0.052 & 0.982 & 0.980 \\ 

6 & 0 & (5, 196, 6, 194, 3, 188, 1, 194, 0, 3, 0, 2, 1, 5, 0, 2) & 0.485 & 0.136 & 0.007 & 33.8\% & 31.7\% & 0.052 & 0.051 & 0.977 & 0.974 \\ 

6 & 1/2 & (16, 266, 1, 129, 0, 126, 0, 247, 0, 0, 0, 2, 0, 4, 0, 9) & 0.481 & 0.092 & 0.009 & 20.7\% & 18.4\% & 0.052 & 0.051 & 0.968 & 0.966 \\ 
\hline
\end{tabular}
\label{tab:simu}
\end{sidewaystable}

\section{Empirical examples}
\label{sec:application}

\subsection{A study on smoking habits}

In 2004, the University of Kansas Medical Center conducted a randomized controlled $2^2$ factorial design to study the smoking habits of African American light smokers, i.e., those ``who smoke 10 or fewer cigarettes per day for at least six months prior to the study'' \citep{Ahluwalia:2006}. The study focused on two treatment factors -- nicotine gum consumption (2gm/day vs. placebo), and counseling  (health education vs. motivational interviewing). Among $N = 755$ participants, $n_1 = 189$ were randomly assigned to $\bm z_1$ (placebo and motivational interviewing), $n_2 = 188$ to $\bm z_2$ (placebo and health education), $n_3 = 189$ to $\bm z_3$ (nicotine gum and motivational interviewing), and $n_4 = 189$ to $\bm z_4$ (nicotine gum and health education). The primary outcome of interest was abstinence from smoking 26 weeks after enrollment, determined by whether salivary cotinine level was less than 20 ng/ml. \cite{Ahluwalia:2006} reported that  
$
(n_1^\mathrm{obs}, n_2^\mathrm{obs}, n_3^\mathrm{obs}, n_4^\mathrm{obs})
= (13, 29, 19, 34).
$

We re-analyze this data set in order to illustrate our proposed methodology. To save space we only focus on $\bar \tau_2,$ the main effect of counseling. The observed data suggests that its point estimate $\hat{\bar \tau}_2 = -0.082,$ the 95\% confidence intervals based on the classic and improved Neymanian variance estimators are (-0.129, -0.035) and (-0.127, -0.037), respectively. While the results largely corroborate \cite{Ahluwalia:2006}'s analysis and conclusion, the improved variance estimator does provide a narrower confidence interval -- the variance estimate by the improved Neymanian variance estimator is 92.1\% of that by the classic Neymanian variance estimator.

\subsection{A study on saphenous-vein coronary-artery bypass grafts}

The Post Coronary Artery Bypass Graft trial is a randomized controlled $2^2$ factorial design conducted between March 1989 and August 1991, on patients who were ``21 to 74 years of age, had low-density lipoprotein (LDL) cholesterol levels of no more than 200 mg/deciliter, and
had had at least two saphenous-vein coronary bypass grafts placed 1 to 11 years before the start of the study'' \citep{Post:1997}. The study concerned two treatment factors -- LDL cholesterol level lowering (aggressive, goal is 60--85 mg/deciliter vs. moderate), and low-dose anticoagulation (1mg warfarin vs. placebo). Among $N = 1351$ participants, $n_1 = 337$ were randomly assigned to $\bm z_1$ (moderate LDL lowering and placebo), $n_2 = 337$ to $\bm z_2$ (moderate LDL lowering and warfarin), $n_3 = 339$ to $\bm z_3$ (aggressive LDL lowering and placebo), and $n_4 = 337$ to $\bm z_4$ (aggressive LDL lowering and warfarin). For the purpose of illustration, we define the outcome of interest as the composite end point (defined as death from cardiovascular or unknown causes, nonfatal myocardial infarction, stroke, percutaneous transluminal coronary angioplasty, or coronary-artery bypass grafting) four years after enrollment. \cite{Post:1997} (in Table 5 and Figure 2, pp. 160) reported that  
\begin{equation*}
n_1^\mathrm{obs} + n_2^\mathrm{obs} = 103,
\quad
n_3^\mathrm{obs} + n_4^\mathrm{obs} = 85,
\quad
n_2^\mathrm{obs} + n_4^\mathrm{obs} = 89,
\quad
n_4^\mathrm{obs} = 68,
\end{equation*}
which implies that
$
(
n_1^\mathrm{obs}, n_2^\mathrm{obs}, 
n_3^\mathrm{obs}, n_4^\mathrm{obs}
)
= (82, 21, 17, 68).
$

We re-analyze the interactive effect $\bar \tau_3.$ The observed data suggests that $\hat{\bar \tau}_3 = 0.166,$ and the 95\% confidence intervals based on the classic and improved Neymanian variance estimators are (0.130, 0.202) and (0.133, 0.200), respectively. Again, the improved Neymanian variance estimator provides a narrower confidence interval, because its variance estimate is only 87.7\% of that by the classic Neymanian variance estimator. Moreover, the results suggest a statistically significant interactive effect between LDL cholesterol lowering and low-dose anticoagulation treatments, which appeared to be absent in \cite{Post:1997}'s original paper.

\section{Concluding remarks}
\label{sec:conclusion}

Motivated by several empirical examples in medical research, in this paper we studied \cite{Dasgupta:2015}'s randomization-based causal inference framework, under which factorial effects are defined as linear contrasts of potential outcomes under different treatment combinations, and the corresponding difference-in-means estimator's only source of randomness is the treatment assignment itself. However, as pointed out by \cite{Aronow:2014}, a long standing challenge faced by such finite-population frameworks is estimating the true sampling variance of the randomization-based estimator. In this paper, we solve this problem and therefore sharpen randomization-based causal inference for $2^2$ factorial designs with binary outcomes, which is not only of theoretical interest, but also arguably the most common and important setting for medical research among all factorial designs. To be more specific, we propose a new variance estimator improving the classic Neymanian variance estimator by \cite{Dasgupta:2015}. The key idea behind our proposed methodology is obtaining the sharp lower bound of the variance of unit-level factorial effects, and using a plug-in estimator for the lower bound. Through several numerical, simulated and empirical examples, we demonstrated the advantages of our new variance estimator.

There are multiple future directions based on our current work. First, although more of theoretical interests, it is possible to extend our methodology to general $2^K$ factorial designs, or even more complex designs such as $3^k$ or fractional factorial designs. Second, we can generalize our existing results for binary outcomes to other scenarios (continuous, time to event, et al.). Third, although this paper focuses on the ``Neymanian'' type analyses, the Bayesian counterpart of causal inference for $2^2$ factorial designs might be desirable. However, it is worth mentioning that, instead of adopting model-based approaches \citep{Simon:1997}, we seek to extend \cite{Rubin:1978}'s and \cite{Ding:2016}'s finite-population Bayesian causal inference framework to factorial designs, which requires a full Bayesian model on the joint distribution of the potential outcomes under all treatment combinations. However, this direction faces several challenges. For example, characterizing the dependence structure in multivariate binary distributions can be extremely complex, as pointed out by \cite{Cox:1972} and \cite{Dai:2013}. Fourth, it would be interesting to explore the potential use of our proposed variance estimator for constructions of non-parametric tests in factorial designs \citep{Solari:2009, Pesarin:2010}. Fifth, it is possible to further improve our variance estimator, by incorporating pre-treatment covariate information. All of the above are our ongoing or future research projects.

\section*{Acknowledgement}
The author thanks Professor Tirthankar Dasgupta at Rutgers University and Professor Peng Ding at UC Berkeley for early conversations which largely motivated this work, and several colleagues at the Analysis and Experimentation team at Microsoft, especially Alex Deng, for continuous encouragement. Thoughtful comments from the Editor-in-Chief Professor Brian Everitt and two anonymous reviewers have substantially improved the quality and presentation of the paper.

\bibliographystyle{apalike}
\bibliography{factorial_binary_final}

%\newpage
\appendix

\section{Proofs of lemmas, theorems and corollaries}
\label{sec:proofs}

\begin{proof}[Proof of Lemma \ref{lemma:variance-taul-formula}]
The proof in large follows \cite[Dasgupta et al.][]{Dasgupta:2015}. By \eqref{eq:factorial-effects}
\begin{align*}
\sum_{i=1}^N \tau_{il}^2 
& = 2^{-2} \sum_{i=1}^N (\bm h_l^\prime \bm Y_i )^2 \\
& = 2^{-2} \sum_{i=1}^N 
\left(
\sum_{j=1}^4 h_{lj}Y_i(\bm z_j)
\right)^2 \\
& = 2^{-2} \sum_{i=1}^N 
\left(
\sum_{j=1}^4 h_{lj}^2 Y^2_i(\bm z_j)
+
\sum_{j \ne j^\prime} h_{lj}h_{lj^\prime} Y_i(\bm z_j)Y_i(\bm z_{j^\prime})
\right) \\
& = 2^{-2}
\left(
\sum_{j=1}^4 h_{lj}^2 \sum_{i=1}^N  Y^2_i(\bm z_j)
+
\sum_{j \ne j^\prime} h_{lj}h_{lj^\prime} \sum_{i=1}^N Y_i(\bm z_j)Y_i(\bm z_{j^\prime})
\right) \\
& = 2^{-2} 
\left(
\sum_{j=1}^4 N_j
+
\sum_{j \ne j^\prime} h_{lj}h_{lj^\prime} N_{jj^\prime}
\right).
\end{align*}
By combining the above and the fact that
\begin{equation*}
S^2(\bar \tau_l) 
= 
(N-1)^{-1}
\left(
\sum_{i=1}^N \tau_{il}^2 - N \bar \tau_l^2 \right),
\end{equation*}
the proof is complete.
\end{proof}

\begin{proof}[Proof of Lemma \ref{lemma:ie-inequality}]
We only prove the case where $l=1$ and 
$
\bm h_l = (-1, -1, 1, 1)^\prime,
$
because other cases ($l=2, 3$) are analogous. We break down \eqref{eq:ie-inequality} to two parts:
\begin{equation}
\label{eq:ie-inequality-1}
\sum_{j=1}^4 N_j
+
\sum_{j \ne j^\prime} h_{1j}h_{1j^\prime} N_{jj^\prime}
\ge
\sum_{j=1}^4 h_{1j} N_l
\end{equation}
and
\begin{equation}
\label{eq:ie-inequality-2}
\sum_{j=1}^4 N_j
+
\sum_{j \ne j^\prime} h_{1j}h_{1j^\prime} N_{jj^\prime}
\ge
- \sum_{j=1}^4 h_{1j} N_l,
\end{equation}
and prove them one by one. It is worth emphasizing that, for the equality in \eqref{eq:ie-inequality} to hold, we only need the equality in either \eqref{eq:ie-inequality-1} or \eqref{eq:ie-inequality-2} to hold.  

To prove \eqref{eq:ie-inequality-1}, note that 
\begin{equation*}
h_{11} = h_{12} = -1,
\quad
h_{13} = h_{14} = 1,
\end{equation*}
and therefore \eqref{eq:ie-inequality-1} is equivalent to
\begin{equation*}
N_{13} + N_{14} + N_{23} + N_{24}
\le
N_1 + N_2
+
N_{34} + N_{12}.
\end{equation*}
We use the inclusion-exclusion principal to prove the above. First, it is obvious that
\begin{equation}
\label{eq:ie-inequality-3-1}
N_{1234} \le N_{12},
\end{equation}
and the equality holds if and only if the set
\begin{equation*}
\{i: Y_i(\bm z_1) = Y_i(\bm z_2) = 1; \; Y_i(\bm z_3) = 0 \; \mathrm{or} \; Y_i(\bm z_4) = 0 \} = \varnothing,  
\end{equation*}
or equivalently
\begin{equation}
\label{eq:ie-inequality-3-2}
\{i: Y_i(\bm z_1) + Y_i(\bm z_2) = 2; \; Y_i(\bm z_3) + Y_i(\bm z_4) < 2 \} = \varnothing .    
\end{equation}
Second, note that
\begin{align}
\label{eq:ie-inequality-3-3}
N_{13} + N_{14} - N_{134}
& = \# \{i: Y_i(\bm z_1) = Y_i(\bm z_3) = 1\} + \# \{i: Y_i(\bm z_1) = Y_i(\bm z_4) = 1\} \nonumber \\
& - \# \{i: Y_i(\bm z_1) = Y_i(\bm z_3) = Y_i(\bm z_4) = 1\} \nonumber \\
& = \# \{i: Y_i(\bm z_1) = 1; \; Y_i(\bm z_3) =1 \; \mathrm{or} \; Y_i(\bm z_4) = 1 \} \nonumber \\
& \le \# \{i: Y_i(\bm z_1) = 1 \} \nonumber \\
& = N_1.
\end{align}
The equality in \eqref{eq:ie-inequality-3-3} holds if and only if
\begin{equation}
\label{eq:ie-inequality-3-4}
\{i: Y_i(\bm z_1) = 1; \; Y_i(\bm z_3) + Y_i(\bm z_4) = 0 \} = \varnothing.    
\end{equation}
Third, by the same argument we have
\begin{equation}
\label{eq:ie-inequality-3-5}
N_{23} + N_{24} - N_{234} \le N_2,
\end{equation}
and the equality in \eqref{eq:ie-inequality-3-5} holds if and only if
\begin{equation}
\label{eq:ie-inequality-3-6}
\{i: Y_i(\bm z_2) = 1; \; Y_i(\bm z_3) + Y_i(\bm z_4) = 0 \} = \varnothing.    
\end{equation}
Fourth, by applying the similar logic, we have
\begin{align}
\label{eq:ie-inequality-3-7}
N_{134} + N_{234} - N_{1234} \le N_{34},
\end{align}
and the equality in \eqref{eq:ie-inequality-3-7} holds if and only if
\begin{equation}
\label{eq:ie-inequality-3-8}
\{i: Y_i(\bm z_1) + Y_i(\bm z_2) = 0; \; Y_i(\bm z_3) + Y_i(\bm z_4) = 2 \} = \varnothing.    
\end{equation}
By combining \eqref{eq:ie-inequality-3-1}, \eqref{eq:ie-inequality-3-3}, \eqref{eq:ie-inequality-3-5} and \eqref{eq:ie-inequality-3-7}, we have proved that \eqref{eq:ie-inequality-1} holds. Moreover, the equality in \eqref{eq:ie-inequality-1} holds if and only if \eqref{eq:ie-inequality-3-1}, \eqref{eq:ie-inequality-3-3}, \eqref{eq:ie-inequality-3-5} and \eqref{eq:ie-inequality-3-7} hold simultaneously, i.e., the four conditions in \eqref{eq:ie-inequality-3-2}, \eqref{eq:ie-inequality-3-4}, \eqref{eq:ie-inequality-3-6} and \eqref{eq:ie-inequality-3-8} are met simultaneously. We leave it to the readers to verify that this is indeed 
equivalent to \eqref{eq:ie-inequality-hold-1}, i.e. for all $j = 1, \ldots, N,$
\begin{equation*}
Y_i(\bm z_3) + Y_i(\bm z_4) - 1 \le Y_i(\bm z_1) + Y_i(\bm z_2) \le Y_i(\bm z_3) + Y_i(\bm z_4).
\end{equation*}

The proof of \eqref{eq:ie-inequality-2} is symmetrical, because it is equivalent to
\begin{equation*}
N_{13} + N_{14} + N_{23} + N_{24}
\le
N_3 + N_4
+
N_{34} + N_{12}.
\end{equation*}
In particular, the equality in \eqref{eq:ie-inequality-2} holds if and only if \eqref{eq:ie-inequality-hold-2} holds (again we leave the verification to the readers). The proof is complete. 
\end{proof}

\begin{proof}[Proof of Theorem \ref{thm:variance-taul-lower-bound}] 
The proof directly follows from \eqref{eq:factorial-effects}, and Lemma \ref{lemma:variance-taul-formula} and \ref{lemma:ie-inequality}.
\end{proof}

\section{Additional simulation studies}
\label{sec:simu-add}

We conduct an additional series of simulation studies to take into account data generation mechanisms different from those described in Section \ref{sec:simu}. In order to generate a ``diverse'' set of joint distributions of the potential outcomes
$
\bm D = (D_{0000}, D_{0001}, \ldots, D_{1111})
$
while keeping the simulation cases closer to our empirical examples, we let
\begin{equation*}
\lambda_1 = 30, \lambda_j \stackrel{\mathrm{iid.}}{\sim} \mathrm{Unif} (0, 1)
\quad
(j = 2, \ldots, 16);
\quad
\bm p = (\lambda_1, \ldots, \lambda_{16})^\prime \big/ \sum_{j=1}^{16} \lambda_j,     
\end{equation*}
and
\begin{equation*}
\bm D \mid \bm p \sim \mathrm{Multinomial} (800, \bm p).
\end{equation*} 
The main rationale behind the above data generation mechanism is that, in many medical studies the (potential) primary endpoint (e.g., mortality) is zero for most patients under any treatment combination. Indeed, our setting guarantees that on average 66.7\% of the experimental units have $Y_i(\bm z_j) = 0$ for all $j = 1, \ldots, 4.$

We use the aforementioned data generation mechanism to produce 50 simulation cases. For each simulation case, we follow the procedure described in Section \ref{sec:simu}, and (to make the article concise) report only the coverage results in Figure \ref{fg:simu-add}. The results largely agree with the conclusions made in Section \ref{sec:simu}, i.e., the improved Neymanian variance estimator in \eqref{eq:factorial-effects-variance-estimator-improved} always, and sometimes greatly, mitigates the over-estimation issue of the classic Neymaninan variance estimator.

\begin{figure}[htbp]
\centering
\includegraphics[width = .75\linewidth]{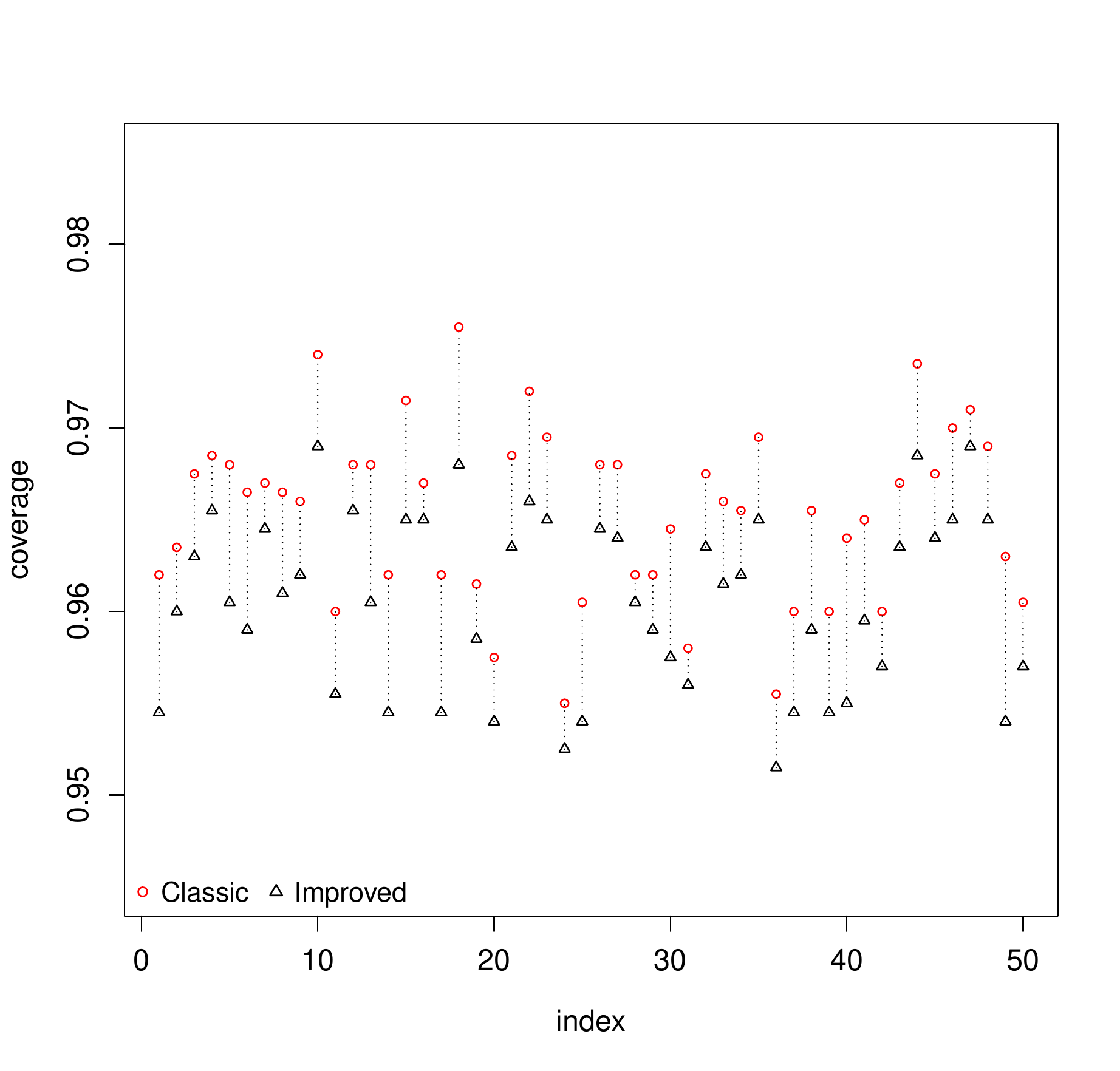}
\caption{Additional simulation results: The horizontal axis contains the indices of the simulation cases, and the vertical shows the coverage rates for the 95\% classic (red rounded dot) and improved (black triangular dot) Neymanian intervals.}
\label{fg:simu-add}
\end{figure}

\end{document}